\documentclass[envcountsame, fleqn]{llncs}
\usepackage[utf8]{inputenc}

\usepackage[hyphens]{url}
\usepackage{hyperref}
\usepackage{hypcap}
\usepackage{xspace}
\usepackage{amsmath}
\usepackage{amssymb}
\usepackage{multicol}
\usepackage{braket}
\usepackage{mathtools}
\usepackage{enumerate}
\usepackage{tikz}
\usetikzlibrary{chains}
\usetikzlibrary{calc}

\newcommand{\redu}[1]{\protect\ensuremath{\leq^{\textnormal{#1}}}}

\makeatletter
\newcommand*{\defeq}{\mathrel{%
  \rlap{\raisebox{0.3ex}{$\m@th\cdot$}}%
  \raisebox{-0.3ex}{$\m@th\cdot$}}%
  =}\makeatother
\makeatletter
\newcommand*{\eqdef}{=\mathrel{%
  \raisebox{0.3ex}{$\m@th\cdot$}%
  \llap{\raisebox{-0.3ex}{$\m@th\cdot$}}}%
  }\makeatother
\makeatletter
\newcommand*{\defeqv}{\mathrel{%
  \rlap{\raisebox{0.3ex}{$\m@th\cdot$}}%
  \raisebox{-0.3ex}{$\m@th\cdot$}}%
  \Longleftrightarrow}\makeatother
\makeatletter
\newcommand*{\eqvdef}{\Longleftrightarrow\mathrel{%
  \raisebox{0.3ex}{$\m@th\cdot$}%
  \llap{\raisebox{-0.3ex}{$\m@th\cdot$}}}%
  }\makeatother

\newcommand{\N}{\protect\ensuremath{\mathbb{N}}\xspace}
\newcommand\classFont[1]{\textnormal{#1}}
\renewcommand{\P}{\protect\ensuremath{\classFont{P}}\xspace}
\newcommand{\NP}{\protect\ensuremath{\classFont{NP}}\xspace}
\newcommand{\NumP}{\protect\ensuremath{{\classFont{\#P}}}\xspace}

\newcommand{\AC}{\protect\ensuremath{{\classFont{AC}^0}}\xspace}
\newcommand{\NC}{\protect\ensuremath{{\classFont{NC}^1}}\xspace}
\newcommand{\FOarb}{\protect\ensuremath{\classFont{FO[Arb]}}\xspace}
\newcommand{\FO}{\protect\ensuremath{\classFont{FO}}\xspace}
\newcommand{\NumAC}{\protect\ensuremath{{\classFont{\#AC}^0}}\xspace}
\newcommand{\WinFOarb}{\protect\ensuremath{\classFont{\#Win-FO[Arb]}}\xspace}
\newcommand{\SkolemFO}{\protect\ensuremath{\classFont{\#Skolem-FO}}\xspace}
\newcommand{\SkolemFOarb}{\protect\ensuremath{\classFont{\#Skolem-FO[Arb]}}\xspace}
\newcommand{\WinFO}{\protect\ensuremath{\classFont{\#Win-FO}}\xspace}
\newcommand{\WinFOPT}{\protect{\classFont{\#Win-FO}[\PLUS,\TIMES]}\xspace}

\newcommand{\FOWinFOarb}{\protect\ensuremath{\classFont{FOCW[Arb]}}\xspace}
\newcommand{\TC}{\protect\ensuremath{\classFont{TC}^0}\xspace}
\newcommand{\CWin}{\protect\ensuremath{\classFont{\#Win}}}
\newcommand{\struc}{\protect\ensuremath{\textrm{STRUC}}}
\newcommand{\tString}{\protect\ensuremath{\tau_{\textrm{string}}}\xspace}
\newcommand{\tCirc}{\protect\ensuremath{\tau_{\textrm{circ}}}}
\newcommand{\PLUS}{+}
\newcommand{\TIMES}{\times}
\newcommand{\BIT}{\mathrm{BIT}}
\newcommand{\tu}[1]{\overline{#1}}

\newcommand{\eg}{e.g.\@\xspace}
\newcommand{\ST}{such that\@\xspace}

\newcommand{\fa}{for all\@\xspace}
\newcommand{\stfa}{such that for all\@\xspace}
\newcommand{\wLOG}{without loss of generality\@\xspace}
\newcommand{\WLOG}{Without loss of generality\@\xspace}

\begin{document}

\title{A Model-Theoretic Characterization of Constant-Depth Arithmetic Circuits}

\author{Anselm Haak \and Heribert Vollmer}
\institute{Theoretische Informatik, Leibniz Universität Hannover,\\Appelstraße, D-30167, Germany\\\email{(haak|vollmer)@thi.uni-hannover.de}}

\maketitle

\begin{abstract}
We study the class $\NumAC$ of functions computed by constant-depth polynomial-size arithmetic circuits of unbounded fan-in addition and multiplication gates.
No model-theoretic characterization for arithmetic circuit classes is known so far. Inspired by Immerman's characterization of the Boolean circuit class \AC, we remedy this situation and develop such a characterization of \NumAC. Our characterization can be interpreted as follows: Functions in \NumAC are exactly those functions counting winning strategies in first-order model checking games. A consequence of our results is a new model-theoretic characterization of \TC, the class of languages accepted by constant-depth polynomial-size majority circuits.
\end{abstract}

\section{Introduction}

Going back to questions posed by Heinrich Scholz and Günter Asser in the early 1960s, Ronald Fagin \cite{FaginThm} laid the foundations for the areas of finite model theory and descriptive complexity theory. He characterized the complexity class \NP as the class of those languages that can be defined in predicate logic by existential second-order sentences: $\NP=\mathrm{ESO}$. His result is the cornerstone of a wealth of further characterizations of complexity classes, cf.~the monographs \cite{EbbFlumThomBuch,ImmermanBuch,LibkinBuch}.

Fagin's Theorem has found a nice generalization: Considering first-order formulae with a free relational variable, instead of asking if there \emph{exists} an assignment to this variable that makes the formula true (ESO), we now ask to \emph{count} how many assignments there are. In this way, the class \NumP is characterized: $\NumP=\#\FO$ \cite{DescCompNumP}. 

But also ``lower'' complexity classes, defined using families of Boolean circuits, have been considered in a model-theoretical way. Most important for us is the characterization of the class $\AC$, the class of languages accepted by families of Boolean circuits of unbounded fan-in, polynomial size and constant depth, by first-order formulae. This correspondence goes back to Immerman and his co-authors \cite{imm87,baimst90}, but was somewhat anticipated by \cite{gule84}. Informally, this may be written as $\AC=\FO$; and there are  two ways to make this formally correct---a non-uniform one: $\AC=\FOarb$, and a uniform one: $\FO[\PLUS, \TIMES]\text{-uniform }\AC=\FO[\PLUS,\TIMES]$ (for details, see below).

In the same way as $\NumP$ can be seen as the counting version of $\NP$, there is a counting version of $\AC$, namely $\NumAC$, the class of those functions counting proof-trees of $\AC$-circuits. A proof-tree is a minimal sub-circuit of the original circuit witnessing that it outputs $1$. Equivalently, $\NumAC$ can be characterized as those functions computable by polynomial-size constant-depth circuits with unbounded fan-in $\PLUS$ and $\TIMES$ gates (and Boolean inputs); for this reason we also speak of \emph{arithmetic} circuit classes. 

For such arithmetic classes, no model-theoretic characterization is known so far. Our \emph{rationale} is as follows: A Boolean circuit accepts its input if it has at least one proof-tree. An \FO-formula (w.l.o.g.~in prenex normal form) holds for a given input if there are Skolem functions determining values for the existentially quantified variables, depending on those variables quantified to the left. By establishing a one-one correspondence between proof-trees and Skolem functions, we show that the class \NumAC, defined by {counting} proof-trees, is equal to the class of functions counting Skolem functions, or, alternatively, winning-strategies in first-order model-checking games: $\AC = \SkolemFO = \WinFO$. We prove that this equality holds in the non-uniform as well as in the uniform setting.

It seems a natural next step to allow first-order formulae to ``talk'' about winning strategies, i.e., allow access to \WinFO-functions (like to an oracle). We will prove that in doing so, we obtain a new model-theoretic characterization of the circuit class \TC of polynomial-size constant-depth MAJORITY circuits.

This paper is organized as follows: In the upcoming section, we will introduce the relevant circuit classes and logics, and we state characterizations of the former by the latter known from the literature. We will also recall arithmetic circuit classes and define our logical counting classes $\SkolemFO$ and $\WinFO$. Sect.~\ref{sect:non-uniform} proves our characterization of non-uniform $\NumAC$, while Sect.~\ref{sect:uniform} proves our characterization of uniform $\NumAC$. Sect.~\ref{sect:threshold} presents our new characterization of the circuit class $\TC$. Finally, Sect.~\ref{sect:concl} concludes with some open questions.

\section{Circuit Classes, Counting Classes, and Logic}
\label{sect:prelim}

\subsection{Non-uniform Circuit Classes}
\label{subsect:nu-circ}

A relational vocabulary is a tuple $\sigma=(R_1^{a_1}, \dots, R_k^{a_k})$, where $R_i$ are relation symbols and $a_i$ their arities, $1\leq i\leq k$. We define first-order formulae over $\sigma$ as usual (see, \eg, \cite{EbbFlumThomBuch,ImmermanBuch}).
First-order structures fix the set of elements (the universe) as well as interpretations for the relation symbols in the vocabulary. Semantics is defined as usual. For a structure $\mathcal{A}$, $|\mathcal{A}|$ denotes its universe. We only consider finite structures here, which means their universes are finite.

Since we want to talk about languages accepted by Boolean circuits, we will use the \emph{vocabulary}
\[\tString \defeq (\leq^2, S^1)\]
\emph{of binary strings}. A binary string is represented as a structure over this vocabulary as follows: Let $w \in \{0,1\}^*$ with $|w| = n$. Then the structure representing this string has universe $\{0,\dots,n-1\}$, $\leq^2$ is interpreted as the $\leq$-relation on the natural numbers and $x \in S$, iff the $x$-th bit of $w$ is 1. The structure corresponding to string $w$ will be called $\mathcal{A}_w$. Vice versa, structure $\mathcal{A}_w$ is simply encoded by $w$ itself: The bits define which elements are in the $S$-relation---the universe and the order are implicit. This encoding can be generalized to binary encodings of arbitrary $\sigma$-structures $\mathcal{A}$. We will use the notation $\textrm{enc}_\sigma(\mathcal{A})$ for such an encoding.

A Boolean circuit $C$ is a directed acyclic graph (dag), whose nodes (also called gates) are marked with either a Boolean function (in our case $\land$ or $\lor$) or a (possibly negated) query of a particular position of the input. Also, one gate is marked as the output gate. A circuit computes a Boolean function on its input bits by evaluating all gates according to what they are marked with. The value of the output gate gives then the result of the computation of $C$ on $x$. We will denote the function computed by circuit $C$ simply by $C$.

A single circuit computes only a finite Boolean function.
When we want circuits to work on different input lengths, we have to consider families of circuits. A family contains one circuit for any input length $n \in \mathbb{N}$. Families of circuits allow us to talk about languages being accepted by circuits. A circuit family $\mathcal{C} = (C_n)_{n \in \mathbb{N}}$ is said to accept (or decide) the language $L$, if it computes its characteristic function $c_L$:
\[C_{|x|}(x) = c_L(x) \textrm{ for all } x.\]

Since we will describe Boolean circuits by FO-formulae, we define the vocabulary
\[\tCirc \defeq (E^2, G_\land^1, G_\lor^1, \textrm{Input}^2, \textrm{negatedInput}^2, r^1),\]
the \emph{vocabulary of Boolean circuits}. The relations are interpreted as follows:
\begin{itemize}
  \item $E(x,y)$:  $y$ is a child of $x$
  \item $G_\land(x)$: gate $x$ is an and-gate
  \item $G_\lor(x)$: gate $x$ is an or-gate
  \item $\textrm{Input}(x, i)$: the $i$-th input is associated with gate $x$
  \item $\textrm{negatedInput}(x, i)$: the negated $i$-th input is associated with gate $x$
  \item $r(x)$: $x$ is the root of the circuit
\end{itemize}

The definition from \cite{ImmermanBuch} is more general in that it allows negations to occur arbitrary in a circuit. Here we only consider circuits in \emph{negation normal form}, i.e., negations are only applied to input bits. This restriction is customary for arithmetic circuits like for the class \NumAC to be defined below. On the other hand, we associate leaves with input positions instead of having a predicate directly determining the truth value of leaves. Still, the classes defined from these stay the same.

The complexity classes in circuit complexity are classes of languages that can be decided by circuit families with certain restrictions on their depth or size. The depth here is the length of a longest path from any input gate to the output gate of a circuit and the size is the number of non-input gates in a circuit. Depth and size of a circuit family are defined as functions accordingly.

\begin{definition}
The class \AC is the class of all languages decidable by Boolean circuit families of constant depth and polynomial size.
\end{definition}

In this definition we do not have any restrictions on the computability of the function $n\mapsto\langle C_n\rangle$, i.e., the function computing (an encoding of) the circuit for a given input length. This phenomenon is referred to as \emph{non-uniformity}, and it leads to undecidable problems in \AC. In first-order logic there is a class that has a similar concept, the class \FOarb, to be defined next.

For arbitrary vocabularies $\tau$ disjoint from $\tString$, we consider formulae over $\tString \cup \tau$ and our input structures will always be \tString-structures $\mathcal{A}_w$ for a string $w \in \{0,1\}^*$. Here, $\tString \cup \tau$ denotes the vocabulary containing all symbols from both vocabularies without repetitions. To evaluate a formula we additionally specify a (non-uniform) family $I = (I_n)_{n\in \mathbb{N}}$ of interpretations of the relation symbols in $\tau$. For $\mathcal{A}_w$ and $I$ as above we now evaluate $\mathcal{A}_w \vDash_I \varphi$ by using the universe of $\mathcal{A}_w$ and the interpretations from both $\mathcal{A}_w$ and $I_{|w|}$. The language defined by a formula $\varphi$ and a family of interpretations $I$ is
\[L_I(\varphi) \defeq \{w\in \{0,1\}^* \mid \mathcal{A}_w \vDash_I \varphi\}\]
This leads to the following definition of \FOarb (equivalent to the one given in \cite{HeribertBuch}):

\begin{definition}\label{def:FOarb}
A language $L$ is in \FOarb, if there are an arbitrary vocabulary $\tau$ disjoint from $\tString$, a first-order sentence $\varphi$ over $\tString \cup \tau$ and a family $I = (I_n)_{n\in \mathbb{N}}$ of interpretations of the relation symbols in $\tau$ \ST
\[L_I(\varphi) = L.\]
\end{definition}

It is known that the circuit complexity class \AC and the model theoretic class \FOarb are in fact the same (see, \eg, \cite{HeribertBuch}):
\begin{theorem}
$\AC = \FOarb$.
\end{theorem}

\subsection{Uniform Circuit Classes}

As already stated, non-uniform circuits are able to solve undecidable problems, even when restricting size and depth of the circuits dramatically. Thus, the non-uniformity somewhat obscures the real complexity of problems. There are different notions of uniformity to deal with this problem: The computation of the circuit $C_{|x|}$ from $x$ must be possible within certain bounds, \eg polynomial time, logarithmic space, logarithmic time. Since we are dealing with FO-formulae, the most natural type of uniformity to use is first-order uniformity, to be defined in this section.

In the logical languages, ``uniformity'' means we now remove the non-uniform family of interpretations from the definition of \FOarb, and replace it with two special symbols for arithmetic, a ternary relation $\PLUS$ (with the intended interpretation $\PLUS(i,j,k)$ iff $i+j=k$) and a ternary relation $\TIMES$ (with the intended interpretation $\TIMES(i,j,k)$ iff $i\cdot j=k$).

\begin{definition}
A language $L$ is in $\FO[\PLUS,\TIMES]$, if there is a first-order sentence $\varphi$ over $\tString \cup (\PLUS,\TIMES)$ \ST
\[L = L_I(\varphi),\]
where $I$ interprets $\PLUS$ and $\TIMES$ in the intended way.
\end{definition}

In the circuit world, as mentioned, ``uniformity'' means 
we can access from any given input $w$ also the circuit $C_{|w|}$ with limited resources. The way we achieve this is via \emph{FO-interpretations}.

In the following, for any vocabulary $\sigma$, $\struc[\sigma]$ denotes the set of all structures over $\sigma$.

\begin{definition}
Let $\sigma, \tau$ be vocabularies,
$\tau = (R_1^{a_1}, \dots, R_r^{a_r})$,
and let $k \in \mathbb{N}$. A \emph{first-order interpretation} (or \emph{FO-interpretation})
\[I\colon \struc[\sigma] \rightarrow \struc[\tau]\]
is given by a tuple of FO-formulae $\varphi_0, \varphi_1, \dots, \varphi_r$ over the vocabulary $\sigma$. $\varphi_0$ has $k$ free variables and $\varphi_i$ has $k \cdot a_i$ free variables \fa $i \geq 1$. For each structure $\mathcal{A} \in \struc[\sigma]$, these formulae define the structure
\[I(\mathcal{A}) = \bigl( |I(\mathcal{A})|, R_1^{I(\mathcal{A})}, \dots, R_r^{I(\mathcal{A})} \bigr) \in \struc[\tau],\]
where the universe is defined by $\varphi_0$ and the relations are defined by $\varphi_1, \dots, \varphi_r$ in the following way:
\[|I(\mathcal{A})| = \bigl\{(b^1, \dots, b^k) \bigm| \mathcal{A} \vDash \varphi_0(b^1, \dots, b^k)\bigr\} \mathrm{\ and}\]
\[R_i^{I(\mathcal{A})} = \bigl\{(\tu{b}_1, \dots, \tu{b}_{a_i}) \in |I(\mathcal{A})|^{a_i} \mid \mathcal{A} \vDash \varphi_i(\tu{b}_1, \dots, \tu{b}_{a_i})\bigr\},\]
where the $\tu{b}_i$ are tuples with $k$ components.
\end{definition}

The term ``FO-interpretations'' was used, \eg, in \cite{AnujFPL}. Sometimes they are also referred to as first-order queries, see, e.g., \cite{ImmermanBuch}. 
They are not to be confused with interpretations of relation symbols as in Sect.~\ref{subsect:nu-circ}. It is customary to use the same symbol $I$ in both cases.


Analogously, $\FO[\PLUS,\TIMES]$-interpretations are interpretations given by tuples of $\FO[\PLUS,\TIMES]$-formulae.

\begin{definition}
A circuit family $\mathcal{C} = (C_n)_{n \in \mathbb{N}}$ is said to be 
\emph{$\FO[\PLUS,\TIMES]$-uniform} if there is an $\FO[\PLUS,\TIMES]$-interpretation
\[I\colon \struc[\tString] \rightarrow \struc[\tCirc]\]
mapping from an input word $w$ given as a structure $\mathcal{A}_w$ over \tString to the circuit $C_{|w|}$ given as a structure over the vocabulary \tCirc. 
\end{definition}

Now we can define the $\FO[\PLUS, \TIMES]$-uniform version of \AC:
\begin{definition}
$\FO[\PLUS,\TIMES]$-uniform \AC is the class of all languages that can be decided by $\FO[\PLUS,\TIMES]$-uniform \AC circuit families.
\end{definition}

Thus, if $\mathcal{C} = (C_n)_{n \in \mathbb{N}}$ is an $\FO[\PLUS,\TIMES]$-uniform circuit family
, we can define from any given input structure $\mathcal{A}_w$ also the circuit $C_{|w|}$ in a first-order way. 

Alternatively, both in the definition of $\FO[\PLUS, \TIMES]$ and $\FO[\PLUS, \TIMES]$-uniform \AC, we can replace $\PLUS$ and $\TIMES$ by the binary symbol $\BIT$ with the meaning $\BIT(i,j)$ iff the $i$th bit in the binary representation of $j$ is $1$, giving rise to the same classes, see also \cite{ImmermanBuch}.

Interestingly, uniform \AC coincides with \FO with built-in arithmetic, see, \eg, \cite{ImmermanBuch}:

\begin{theorem}
$\FO[\PLUS,\TIMES]$-uniform $\AC = \FO[\PLUS,\TIMES]$.
\end{theorem}

\subsection{Counting Classes}

Building on the previous definitions we next want to define counting classes. The objects counted on circuits are proof trees: A \emph{proof tree} is a minimal subtree showing that a circuit evaluates to true for a given input. For this, we first unfold the circuit into tree shape, and we further require that it is in negation normal form. A proof tree then is a tree we get by choosing for any $\lor$-gate exactly one child and for any $\land$-gate all children, \ST every leaf which we reach in this way is a true literal.

Now, \NumAC is the class of functions that ``count proof trees of \AC circuits'':

\begin{definition}
($\FO[\PLUS,\TIMES]$-uniform) \NumAC is the class that consists of all functions $f \colon \{0,1\}^*\rightarrow\N$ for which there is a  circuit family (an $\FO[\PLUS,\TIMES]$-uniform circuit family) $\mathcal{C} = (C_n)_{n \in \mathbb{N}}$
such that for any $x \in \{0,1\}^*$,
$f(x)$ equals the number of proof trees of $C_{|x|}$ on input $x$.
\end{definition}

It is the aim of this paper to give model-theoretic characterizations of these classes. The only model-theoretic characterization of a counting class that we are aware of is the following:
In \cite{DescCompNumP}, a counting version of $\FO$ was defined, inspired by Fagin's characterization of $\NP$. Functions in this class count assignments to free relational variables in \FO-formulae. However, it is known that $\#\P = \#\FO$, i.e., this counting version of \FO coincides with the much larger counting class $\#\P$ of functions counting accepting paths of nondeterministic polynomial-time Turing machines. It is known that $\NumAC\subsetneq\#\P$.
Thus, we need some weaker notion of counting. 

Suppose we are given a $\tString$-formula $\varphi$ in prenex normal form,
$$\varphi = \exists y_1 \forall z_1 \exists y_2 \forall z_2 \dots \exists y_{k-1} \forall z_{k-1} \exists y_k \,\, \psi(\overline{y}, \overline{z})$$
for quantifier-free $\psi$. If we want to satisfy $\varphi$ in a word model $\mathcal{A}_w$, we have to find an assignment for $y_1$ such that for all $z_1$ we have to find an assignment for $y_2$ \dots such that $\psi$ is satisfied in $\mathcal{A}_w$. Thus, the number of ways to satisfy $\varphi$ consists in the number of picking the suitable $y_i$, depending on the universally quantified variables to the left, such that $\psi$ holds, in other words, the number of Skolem functions for the existentially quantified variables.

\begin{definition}
A function $g\colon \{0,1\}^* \rightarrow \mathbb{N}$ is in the class $\SkolemFOarb$ if there is a vocabulary $\tau$ disjoint from $\tString$, a sequence of interpretations $I = (I_n)_{n \in \mathbb{N}}$ for $\tau$ and a first-order sentence $\varphi$ over $\tString \cup \tau$ in prenex normal form  
\[\varphi = \exists y_1 \forall z_1 \exists y_2 \forall z_2 \dots \exists y_{k-1} \forall z_{k-1} \exists y_k \,\, \psi(\overline{y}, \overline{z})\]
with quantifier-free $\psi$, \stfa $w \in \{0,1\}^*$,
$g(w)$ is equal to the number of tuples $(f_1,\dots,f_k)$ of functions such that 
$$\mathcal{A}_w \ \vDash_I \ \forall z_1 \dots \forall z_{k-1} \ \psi(f_1, f_2(z_1), \dots, f_k(z_1, \dots, z_{k-1}), z_1, \dots, z_{k-1})\}$$
\end{definition}

This means that $\SkolemFOarb$ contains those functions that, for a fixed \FO-formula $\varphi$, map an input $w$ to the number of Skolem functions of $\varphi$ on $\mathcal{A}_w$.

A different view on this counting class is obtained by recalling the well-known game-theoretic approach to first-order model checking. Model checking for FO-formulae (in prenex normal form) can be characterized using a two player game: The verifier wants to show that the formula evaluates to true, whereas the falsifier wants to show that it does not. For each quantifier, one of the players chooses an action: For an existential quantifier, the verifier chooses which element to take (because he needs to prove that there is a choice satisfying the formula following after the quantifier). For a universal quantifier, the falsifier chooses which element to take (because he needs to prove that there is a choice falsifying the formula following after the quantifier). When all quantifiers have been addressed, it is checked whether the quantifier-free part of the formula is true or false. If it is true, the verifier wins. Else, the falsifier wins. Now the formula is fulfilled by a given model, iff there is a winning strategy (for the verifier).

\begin{definition}
A function $f$ is in \WinFOarb, if there are a vocabulary $\tau$ disjoint from $\tString$, a sequence of interpretations $I = (I_n)_{n \in \mathbb{N}}$ for $\tau$ and a first-order sentence $\varphi$ in prenex normal form over $\tString \cup \tau$ \stfa $w \in \{0,1\}^*$,
$f(w)$ equals the number of winning strategies for the verifier in the game for $\mathcal{A}_w \vDash_I \varphi$.
\end{definition}

The correspondence between Skolem functions and winning strategies has been observed in far more general context, see, e.g., \cite{Gradel13}. In our case, this means that
$$\SkolemFOarb=\WinFOarb.$$

Analogously we define the uniform version (which we only state using the notion of the model checking games):

\begin{definition}
A function $f$ is in $\WinFOPT$, if there is a first-order sentence $\varphi$ in prenex normal form over $\tString \cup (\PLUS,\TIMES)$ \stfa $w \in \{0,1\}^*$,
$f(w)$ equals the number of winning strategies for the verifier in the game for $\mathcal{A}_w \vDash_I \varphi$, where $I$ interprets $\PLUS$ and $\TIMES$ in the intended way.
\end{definition}

We will use $\CWin(\varphi, \mathcal{A}, I)$ ($\CWin(\varphi, \mathcal{A})$, resp.) to denote the number of winning strategies for $\varphi$ evaluated on the structure $\mathcal{A}$ and the interpretation $I$ (the structure $\mathcal{A}$ and the intended interpretation of $\PLUS$ and $\TIMES$, resp.). In the previous two definitions we could again have replaced $\PLUS$ and $\TIMES$ by $\BIT$.

In the main result of this paper, we will show that the thus defined logical counting classes equal the previously defined counting classes for constant-depth circuits.

\section{A Model-Theoretic Characterization of \NumAC}
\label{sect:non-uniform}

We first note that there is a sort of closed formula for the number of winning strategies of \FO-formulae on given input structures:

\begin{lemma}\label{lemWinCount}
Let $\tau_1, \tau_2$ be disjoint vocabularies and $I$ an interpretation of $\tau_2$. Let $\varphi$ be an FO-formula in prenex normal form over the vocabulary $\tau_1 \cup \tau_2$
of the form
\[\varphi = Q_1 x_1 \dots Q_n x_n \psi,\]
where $Q_i \in \{\exists, \forall\}$ and $\psi$ is quantifier-free.\\
Let $\mathcal{A}$ be a $\tau_1$-structure and $I$ a sequence of interpretations for the symbols in $\tau_2$%
.
Then the number of winning strategies for $\mathcal{A} \vDash_I \varphi$ 
is the following:
\[\CWin(\varphi, \mathcal{A}, I) = \mathop{\Delta_1}_{a_1 \in |\mathcal{A}|} \mathop{\Delta_2}_{a_2 \in |\mathcal{A}|} \cdots \mathop{\Delta_n}_{a_n \in |\mathcal{A}|}
\begin{cases}
1 & \textrm{, if } \mathcal{A} \vDash_I \varphi(a_1, \dots, a_n)\\
0 & \textrm{, else},
\end{cases}\]
where
$\Delta_i$ is a sum if $Q_i = \exists$ and a product otherwise.
\end{lemma}
Since $\CWin(\varphi, \mathcal{A})$ is defined via $\CWin(\varphi, \mathcal{A}, I)$, the result also applies to the uniform setting.

Our main theorem can now be stated as follows:

\begin{theorem}\label{thm:main}
$\NumAC = \WinFOarb$ 
\end{theorem}

The rest of this section is devoted to a proof of this theorem.

\begin{proof}
\underline{$\subseteq$}: Let $f$ be a function in $\NumAC$ and $\mathcal{C} = (C_n)_{n \in \mathbb{N}}$ an $\AC$ circuit family witnessing this. Assume that all $C_n$ are already trees and all leaves have the same depth (the latter can be achieved easily by adding and-gates with only one input). Also, we can assume that all $C_n$ use and- and or-gates alternating beginning with an and-gate in the root. This can be achieved by doubling the depth of the circuit and replacing each layer of the old circuit by an and-gate followed by an or-gate. In Figure \ref{fig:Alternate}, this replacement is illustrated.
\begin{figure}
\capstart
\begin{center}
\scalebox{.6}{
\begin{tikzpicture}

\node[circle, draw] (p) {$p$};
\draw (p.south) -- ($(p.south) - (0,1)$);

\node [circle, draw] (current)[below= 1 of p, anchor=north] {$\land$};
\draw (current.south west) -- ($(current.south west) - (1.5,1)$);
\draw (current.south east) -- ($(current.south east) + (1.5,-1)$);

\node[circle, draw] (c1)[below left= 1 and 1.5 of current] {$c_1$};
\node[circle, draw] (cn)[below right= 1 and 1.5 of current] {$c_n$};
\node[below=0.5 of current] {$\cdots$};

\node[right=3 of current, scale=2.5] (leads){$\leadsto$};

\node [circle, draw] (currentNew)[right=3 of leads] {$\land$};
\draw (currentNew.north) -- ($(currentNew.north) + (0,1)$);

\node[circle, draw] (pNew)[above= 1 of currentNew] {$p$};

\draw (currentNew.south west) -- ($(currentNew.south west) - (1.5,1)$);
\draw (currentNew.south east) -- ($(currentNew.south east) + (1.5,-1)$);
\node [circle, draw=red] (or1)[below left= 1 and 1.5 of currentNew, color=red] {$\lor$};
\node [circle, draw=red] (orn)[below right= 1 and 1.5 of currentNew, color=red] {$\lor$};
\node[below=0.5 of currentNew] {$\cdots$};

\draw[color=red] (or1) -- ($(or1.south) - (0,1)$);
\draw[color=red] (orn) -- ($(orn.south) - (0,1)$);
\node[circle, draw] (c1New)[below= 1 of or1] {$c_1$};
\node[circle, draw] (cnNew)[below= 1 of orn] {$c_n$};

\node (AND)[left= 3 of current, scale=1.5] {$\land$-gate:};

\node[circle, draw] (p)[below= 7 of p] {$p$};
\draw (p.south) -- ($(p.south) - (0,1)$);

\node [circle, draw] (current)[below= 1 of p, anchor=north] {$\lor$};
\draw (current.south west) -- ($(current.south west) - (1.5,1)$);
\draw (current.south east) -- ($(current.south east) + (1.5,-1)$);

\node[circle, draw] (c1)[below left= 1 and 1.5 of current] {$c_1$};
\node[circle, draw] (cn)[below right= 1 and 1.5 of current] {$c_n$};
\node[below=0.5 of current] {$\cdots$};

\node[right=3 of current, scale=2.5] (leads){$\leadsto$};

\node [circle, draw] (currentNew)[right=3 of leads] {$\lor$};

\draw[color=red] (currentNew.north) -- ($(currentNew.north) + (0,1)$);
\node[circle, draw=red] (and)[above= 1 of currentNew, color=red] {$\land$};

\draw (and.north) -- ($(and.north) + (0,1)$);
\node[circle, draw] (pNew)[above= 1 of and] {$p$};

\draw (currentNew.south west) -- ($(currentNew.south west) - (1.5,1)$);
\draw (currentNew.south east) -- ($(currentNew.south east) + (1.5,-1)$);
\node[circle, draw] (c1New)[below left= 1 and 1.5 of currentNew] {$c_1$};
\node[circle, draw] (cnNew)[below right= 1 and 1.5 of currentNew] {$c_n$};
\node[below=0.5 of currentNew] {$\cdots$};

\node (OR)[left= 3 of current, scale=1.5] {$\lor$-gate:};

\end{tikzpicture}
}
\caption{Construction of an Alternating Circuit}
\label{fig:Alternate}
\end{center}
\end{figure}
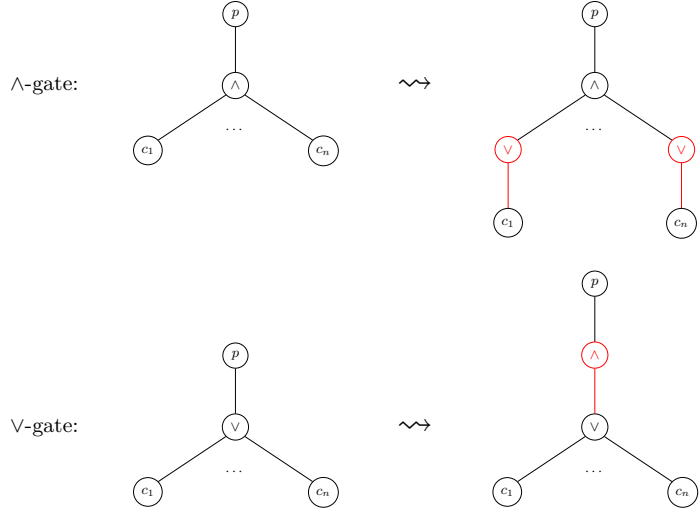


Let $w \in \{0,1\}^*$ be an input, $r$ be the root of $C_{|w|}$ and $k$ the depth of $C_n$ for all $n$. The value $f(w)$ can be given as follows:
\[\label{eqn:prooftrees}f(w) = \quad \mathclap{\prod_{\substack{y_1 \textrm{ is a}\\\textrm{child of } r}}} \quad \quad \sum_{\substack{y_2 \textrm{ child}\\\textrm{of } y_1}} \!\! \cdots \!\! \mathop{\bigcirc}_{\substack{y_k \textrm{ child}\\\textrm{of } y_{k-1}}}
\begin{cases}
1 & \textrm{, if } y_k \textrm{ associated with a true literal}\\
0 & \textrm{, else},
\end{cases}\tag{$\star$}\]
where $\bigcirc$ is a product if $k$ is odd and a sum otherwise.

We will now build an FO-sentence $\varphi$ over $\tString \cup \tCirc$ such that for any input $w \in \{0,1\}^*$, the number of winning strategies to verify $\mathcal{A}_w \vDash_\mathcal{C} \varphi$ equals the number of proof trees of the circuit $C_{|w|}$ on input $w$. Note that the circuit family $\mathcal{C}$ as a family of \tCirc-structures can directly be used as the non-uniform family of interpretations for the evaluation of $\varphi$. Since only one universe is used for evaluation and it is determined by the input structure $\mathcal{A}_w$,  the gates in this \tCirc-structure have to be tuples of variables ranging over the universe of $\mathcal{A}_w$. To simplify the presentation, we assume in the following that we do not need tuples---a single element of the universe already corresponds to a gate. The proof can be generalized to the case where this assumption is dropped.

Before giving the desired sentence $\varphi$ over $\tString \cup \tCirc$, we define
\[\varphi_\textrm{trueLiteral}(x) \defeq \exists i \ \big(\textrm{Input}(x, i) \land S(i) \lor \textrm{negatedInput}(x, i) \land \neg S(i)\big).\]
Now $\varphi$ can be given as follows:
\begin{align*}
\varphi \defeq & \exists y_0 \forall y_1 \exists y_2 \dots Q_k y_k\\
& r(y_0) \land
\left(\vphantom{\bigvee_{\substack{1 \leq i \leq k,\\i \textrm{ odd}}}}\left( \left( \bigwedge_{1 \leq i \leq k}{E(y_i, y_{i-1})}\right) \land 
\varphi_\textrm{trueLiteral}(y_k)\right)\right.\\
& \left. \lor \bigvee_{\substack{1 \leq i \leq k,\\i \textrm{ odd}}}
{\left(\bigwedge_{1 \leq j < i}{(E(y_j, y_{j-1}))} \land \neg E(y_i, y_{i-1}) \land \bigwedge_{i < j \leq k}{r(y_j)}\right)}\right),
\end{align*}
where $Q_k$ is an existential quantifier, if $k$ is odd and a universal quantifier otherwise.

The big disjunction ensures that the counted value when making wrong choices (choosing an element that is either not a gate or not a child of the previous gate) is always the neutral element of the arithmetic operation associated with the respective quantifier: For existential quantifiers, we sum over all possiblities. Thus, having counted value 0 for each wrong choice is fine. For universal quantifiers, we multiply the number of winning strategies for all choices, though. In this case, we need to get value 1 for each wrong choice.

We now need to show that the number of winning strategies for $\mathcal{A}_w \vDash_\mathcal{C} \varphi$ is equal to the number of proof trees of the circuit $C_{|w|}$ on input $w$. For this, let 
\begin{align*}
\varphi^{(n)}(y_1, \dots, y_n) \defeq & Q_{n+1} y_{n+1} \dots Q_k y_k\\
& \left( \bigwedge_{1 \leq i \leq k}{\big(E(y_i, y_{i-1})\big)} \land
\varphi_\textrm{trueLiteral}(y_k)\right) \lor\\
& \bigvee_{\substack{n+1 \leq i \leq k,\\i \textrm{ odd}}}
\left(\bigwedge_{1 \leq j < i}{\big(E(y_j, y_{j-1})\big)} \land \neg E(y_i, y_{i-1}) \land \right.\\
& \phantom{\bigvee_{\substack{n+1 \leq i \leq k,\\i \textrm{ odd}}}
\left( \vphantom{\bigwedge_{1 \leq j < i}} \right. } \left. \bigwedge_{i < j \leq k}{r(y_j)}\right),
\end{align*}
where $Q_{n+1}, \dots, Q_{k-1}$ are the quantifiers preceding $Q_k$.
Note that the start of the index $i$ on the big ``or'' changed compared to $\varphi$. In the following we will use the abbreviation
\[\#w(\psi) = \CWin(\psi, \mathcal{A}_w, \mathcal{C}).\]
We now show by induction that:
\[\#w(\varphi) = \prod_{\substack{y_1 \textrm{ is a}\\\textrm{child of } r}}\quad \sum_{\substack{y_2 \textrm{ is a}\\\textrm{child of } y_1}} \dots \mathop{\bigcirc}\limits_{\substack{y_n \textrm{ is a}\\\textrm{child of } y_{n-1}}}
\#w(\varphi^{(n)}[y_0/r]).\]
Here, $r$ is notation for the root of the circuit, although we formally do not use constants. The replacement of $y_0$ by $r$ is only done for simplicity.\\
Induction basis ($n=0$): The induction hypothesis here simply states
\[\#w(\varphi) = \#w(\varphi^{(0)}[y_0/r]),\]
which holds by definition.\\
Induction step ($n \rightarrow n+1$): We can directly use the induction hypothesis here:
\[\#w(\varphi) = \prod_{\substack{y_1 \textrm{ is a}\\\textrm{child of } r}}\quad \sum_{\substack{y_2 \textrm{ is a}\\\textrm{child of } y_1}} \dots \mathop{\bigcirc}\limits_{\substack{y_n \textrm{ is a}\\\textrm{child of } y_{n-1}}}
\#w(\varphi^{(n)}[y_0/r]),\]
so it remains to show that
\[\#w(\varphi^{(n)}[y_0/r]) = \mathop{\bigcirc}\limits_{\substack{y_{n+1} \textrm{ is a}\\\textrm{child of } y_n}} \#w(\varphi^{(n+1)}[y_0/r])\]
We distinguish two cases: Depending on whether $n+1$ is even or odd, the $(n+1)$-st quantifier is either an existential or a universal quantifier. In the same way all gates of that depth in the circuits from $\mathcal{C}$ are either or- or and-gates.

\textit{Case 1}: $n+1$ is odd, so the $(n+1)$-st quantifier is a universal quantifier. Thus, from $\#w(\varphi^{(n)}[y_0/r])$ we get a $\prod$-operator, which is the same we get for an and-gate in the corresponding circuit. We now need to check over which values of $y_{n+1}$ the product runs:\\
The big conjunction may only be true if $y_{n+1}$ is a child of $y_n$.\\
The big disjunction may become true for values of $y_{n+1}$ which are no children of $y_n$ only if all variables quantified after $y_{n+1}$ are set to $r$ (the choice of $r$ here is arbitrary and was only made because $r$ is the only constant in the vocabulary). This has the purpose to make the counted value 1 in this case. Also, the disjunct for $i=n$ can only be made true if $y_{n+1}$ is not a child of $y_n$, so we can drop it if $y_{n+1}$ is a child of $y_n$. Since for all values of $y_{n+1}$ that are not children of $y_n$ we fix all variables quantified afterwards, we get:
\begin{align*}
\prod_{y_{n+1} \in |\mathcal{A}_w|} \#w(\varphi^{(n+1)}[y_0/r]) & =
\quad \prod_{\mathclap{\substack{y_{n+1} \in |\mathcal{A}_w|,\\y_{n+1} \textrm{is a child of } y_n}}} \#w(\varphi^{(n+1)}[y_0/r]) \quad \cdot \quad
\prod_{\mathclap{\substack{y_{n+1} \in |\mathcal{A}_w|,\\ y_{n+1} \textrm{is not a child of } y_n}}} 1\\
& = \quad \prod_{\mathclap{\substack{y_{n+1} \in |\mathcal{A}_w|,\\y_{n+1} \textrm{is a child of } y_n}}} \#w(\varphi^{(n+1)}[y_0/r]).
\end{align*}
Thus, we get a product only over the children of $y_n$ and can drop the disjunct for $i=n$ from the formula for the next step.

\textit{Case 2}: $n+1$ is even, so the $(n+1)$-st quantifier is an existential quantifier. Therefore, we get a $\sum$-operator from $\#w(\varphi^{(n)}[y_0/r])$, which is the same we get for an or-gate in the corresponding circuit. We now need to check over which values of $y_{n+1}$ the sum runs:\\
The big conjunction can only be true if $y_{n+1}$ is a child of $y_n$.\\
The big disjunction can also only be true if $y_{n+1}$ is a child of $y_n$.\\
This means that the counted value, when choosing gates that are not children of the current gate, will be 0. Thus, we directly get the sum
\[\sum_{\substack{y_{n+1} \in |\mathcal{A}_w|,\\y_{n+1} \textrm{ is a child of } y_n}} \#w(\varphi^{(n+1)}[y_0/r]).\]
Here, $\varphi^{(n+1)}$ does not drop a disjunct. This concludes the induction.\\
For $\varphi^{(k)} = \varphi_\textrm{trueLiteral}(y_k)$, we get
\[\#w(\varphi^{(k)}) = \CWin(\varphi^{(k)}, \mathcal{A}_w, \mathcal{C}),\]
which is 1 if $y_k$ is associated with a true literal in $C_{|w|}$ on input $w$ and 0, otherwise. This is due to the fact that the quantifier in $\varphi_\textrm{trueLiteral}$ can only be made true by exactly one choice of $i$. Together with equation (\ref{eqn:prooftrees}) on page \pageref{eqn:prooftrees} this shows $f(w) = \#w(\varphi)$.

\underline{$\supseteq$}: Let $f$ be a function in \WinFOarb. Let $\tau$ disjoint from $\tString$ be a vocabulary and the formula $\varphi$ over $\tau \cup \tString$ together with the non-uniform family  $I = (I_n)_{n \in \mathbb{N}}$ of interpretations of the relation symbols in $\tau$ a witness for $f \in \WinFOarb$. Let $k$ be the length of the quantifier prefix of $\varphi$. We now describe how to construct the circuits $C_n$ within a circuit family $\mathcal{C} = (C_n)_{n \in \mathbb{N}}$ that shows $f \in \NumAC$.

$C_n$ consist of two parts. The first part mimics the quantifiers of the formula, while the second part is of constant size and evaluates the quantifier-free part of the formula.

The first part is built analogously to the circuit in Immerman's proof of $\FO[\PLUS, \TIMES] \subseteq \FO[\PLUS, \TIMES]\textrm{-uniform } \AC$ \cite{ImmermanBuch}. The gates are of the form $(a_1, \dots, a_i)$ with $1 \leq i \leq k$ and $a_j \in \{0, \dots, n-1\}$ \fa $j$. Each such gate has the meaning that we set the first $i$ quantified variables to the values $a_1, \dots, a_i$. Therefore, for the $i$-th quantifier of $\varphi$ and for any choice of $a_1, \dots, a_{i-1}$, $(a_1, \dots, a_{i-1})$ is an and-gate if the quantifier was $\forall$ and an or-gate if the quantifier was $\exists$. Also, if $i \leq k$ we add as children to each such gate $(a_1, \dots, a_{i-1}, a_i)$ \fa $a_i \in |\mathcal{A}_w|$. This means that the name of each gate in depth $k$ contains information about all choices made for the quantified variables.

The second part, which evaluates the quantifier-free part of $\varphi$ based on the choices made for the quantified variables works as follows: First, transform the quantifier-free part of $\varphi$ into disjunctive normal form. Then expand the clauses to Minterms. For this, think of the formula as a propositional formula which has all occuring atoms as its variables. This means that after the transformation we have an \FO-formula $\varphi'$ in disjunctive normal form, in which all occuring atoms also occur in each clause. The result is that for each assignment either no clause of the formula is satisfied or exactly one clause is satisfied. We can now directly build a circuit computing this disjunctive normal form except for the truth values of the atoms. From the above we get that this circuit is satisfied by an input if and only if it has exactly one proof tree with that input. This means that when counting proof trees, this second part will only determine the truth value of the quantifier-free part and not give a value above 1. The truth value of the atoms can be determined in the circuit as follows: For predicates of the form $S(x)$ use an input gate associated with the $x$-th input bit. All other predicates are numerical predicates and thus only dependent on $n$. Therefore, they can be replaced by constants in $C_n$.

%
Now by Lemma \ref{lemWinCount} it is clear that counting proof trees on this circuit family leads to the same function as counting winning strategies of the verifier for $\mathcal{A}_w \vDash_I \varphi$. 
\end{proof}

\section{The Uniform Case}
\label{sect:uniform}

Next we want to transfer this result to the uniform setting. In the direction from right to left we will have to show that the constructed circuit is uniform, which is straightforward. On the other hand, the following important point changes in the direction from left to right: We have to actually replace  queries to $C_{|w|}$ in the FO-sentence by the corresponding FO-formulae we get from the FO-interpretation which shows uniformity of $\mathcal{C}$. Since we introduce new quantifiers by this, we have to show how we can keep the counted value the same.
That this is possible follows from the following lemma, which can also be used to prove that $\WinFOPT$ is closed under $\FO[\PLUS, \TIMES]$-reductions (exact definition follows).

\begin{lemma}\label{lemAppFOint}
Let $\varphi$ be an $\FO[\PLUS,\TIMES]$-formula over some vocabulary $\tau$, and let $I\colon \struc[\sigma] \rightarrow \struc[\tau]$ be an $\FO[\PLUS,\TIMES]$-interpretation. Then there is an $\FO[\PLUS,\TIMES]$-formula $\varphi'$ over $\sigma$ \stfa $\mathcal{A} \in \struc[\sigma]$,
\[\CWin(\varphi', \mathcal{A}) = \CWin(\varphi, I(\mathcal{A})).\]
\end{lemma}

\begin{proof}
The idea is to plug in the formulae from $I$ into the formula $\varphi$, replacing all occurrences of symbols from $\tau$. Let $\varphi_\textrm{uni}$ be the formula in $I$, which checks membership in the universe of the structure we map to. To simplify the presentation, we assume in the following that we do not need tuples---elements of the universe of $I(\mathcal{A})$ can be encoded as elements of the universe of $\mathcal{A}$. \WLOG we can assume that all existential quantifiers in formulae in $I$ can only be satisfied by at most one witness. If this is not the case, we can replace $\exists$-quantifiers (starting from the outermost) in the following way (forming a NNF in every step):
\[\exists z \psi(z) \leadsto \exists z (\psi(z) \land \forall z_2 (z_2 < z \rightarrow \neg \psi(z_2))),\]
with a fresh variable $z_2$.\\
We now get the desired formula $\varphi'$ as follows:
\begin{enumerate}
  \item Replace every occurrence of a relation symbol from $\tau$ by the corresponding FO-formula from $I$.
  \item Transform the formula into prenex normal form by shifting all newly introduced quantifiers directly behind the old quantifier prefix (renaming variables where neccessary).
  \item Make sure that only tuples that satisfy $\varphi_\textrm{uni}$ are counted (see below).
\end{enumerate}
Step 3 needs a bit of further explanation. Let $\psi$ be the formula after step 2:
\[\psi \defeq Q_1 x_1 \dots Q_k x_k Q_1 y_1 \dots Q_\ell y_\ell \ \psi',\]
where $\psi'$ is quantifier-free. By the assumption above, each $x_i$ corresponds to a variable in $\varphi$. The variables $y_i$ are the variables newly introduced by the formulae from $I$. We now make sure that we count only over tuples from the universe. For this, we start by defining
\[\psi'' \defeq \psi \land \bigwedge_{1\leq i \leq k}{\varphi_\textrm{uni}(x_i)}.\]
The formula $\varphi''$ can only be made true when choosing values for $x_i$ that correspond to elements of $|I(\mathcal{A})|$. When counting winning strategies, there remains a problem, though: The quantifier-free part will be false, whenever values not corresponding to elements of $|I(\mathcal{A})|$ are chosen. This is fine for existential quantifiers, because 0 is the neutral element of addition, but it does not work for universal quantifiers. Similar to the respective part of the proof of Theorem \ref{thm:main}, we need to fix this. We do this by building the formula
\[\varphi' \defeq \psi'' \lor \bigvee_{\substack{1 \leq i < k\\Q_i = \forall}} \left( \bigwedge_{1 \leq j < i} \varphi_\textrm{uni}(x_j) \land \neg \varphi_\textrm{uni}(x_i) \land \bigwedge_{i < j \leq k} \forall z (x_j < z \lor x_j = z) \right).\]
The new part allows universally quantified variables to take on values that do not corrspond to elements of $|I(\mathcal{A})|$, but if they do, all variables quantified afterwardrs are fixed to take a specific value. Thus, for every such choice of a universally quantified variable, the number of winning strategies for the rest of the game is always exactly 1, the neutral element of multiplication (cf. proof of Theorem \ref{thm:main}).

Now let $\varphi''$ be the prenex normal form of $\varphi'$. Since the elements outside the universe are now handled, we get by definition of \FO-interpretations (and thus $I$):\\
The quantifier-free part of $\varphi$ is true for an assignment to the quantified variables if and only if the part of $\varphi''$ after quantification of the variables $x_i$ is true for the corresponding assignment to these variables.\\
Since the existential quantifiers occuring afterwards can have at most one witness, the number of winning strategies for this part can only be 0 or 1. Thus, they can be viewed as normal quantifiers (only determining a truth value) instead of quantifiers for which we count assignments to the quantified variables. This proves
\[\CWin(\varphi', \mathcal{A}) = \CWin(\varphi, I(\mathcal{A})).\]
\end{proof}

As already mentioned, this lemma yields an interesting closure property as a corollary, that is, closure under \FO-reductions:

\begin{definition}
Let $f,g: \{0,1\}^* \rightarrow \mathbb{N}$. 
We say that $f$ is (many-one) $\FO[\PLUS, \TIMES]$-reducible to $g$, 
in symbols: $f \redu{fo} g$, 
if there are vocabularies $\sigma, \tau$ and an $\FO[\PLUS,\TIMES]$-interpretation $I: \struc[\sigma] \rightarrow \struc[\tau]$ \stfa $\mathcal{A} \in \struc[\sigma]$:
\[f(\textrm{enc}_\sigma(\mathcal{A})) = g(\textrm{enc}_\tau(I(\mathcal{A}))).\]
\end{definition}

\begin{corollary}\label{corClosRed}
On ordered structures with $\PLUS$ and $\TIMES$, $\WinFO[\PLUS, \TIMES]$ is closed under $\FO[\PLUS, \TIMES]$-reductions, that is: Let $f,g$ be functions with $g \in \WinFO[\PLUS, \TIMES]$ and $f \redu{fo} g$. Then $f \in \WinFO[\PLUS, \TIMES]$.
\end{corollary}

\begin{proof}
Let $I: \struc[\sigma] \rightarrow \struc[\tau]$ be a witness for $f \redu{fo} g$. Since we have ordered structures with $\PLUS$ and $\TIMES$, we can assume \wLOG that $\sigma = \tau = \tString \cup (\PLUS, \TIMES)$
.\\
Let $\varphi$ over $\tString \cup (\PLUS, \TIMES)$ be a witness for $g \in \WinFO[\PLUS, \TIMES]$, so we have \fa $\mathcal{A}_w \in \struc[\tString]$:
\[\CWin(\varphi, \mathcal{A}_w) = g(w).\]
Now we use Lemma \ref{lemAppFOint} to get $\varphi'$ over vocabulary $\tString \cup (\PLUS, \TIMES)$ with
\[\CWin(\varphi', \mathcal{A}) = \CWin(\varphi, I(\mathcal{A}))\]
and get
\begin{align*}
\CWin(\varphi', \mathcal{A}_w) & = \CWin(\varphi, I(\mathcal{A}_w))\\
& = g(\textrm{enc}_{\tString}(I(\mathcal{A}_w)))\\
& = f(\textrm{enc}_{\tString}(\mathcal{A}_w))\\
& = f(w).
\end{align*}
\end{proof}

\begin{remark}
The result does not hold in this simple form if $\PLUS$ and $\TIMES$ are not part of $\tau$. This is due to the fact, that the numerical predicates of the structure we map to might not be definable. They need to be given within the \FO-interpretation $I$. Alternatively, we could restrict ourselves to \FO-interpretations that have $\varphi_\textrm{universe} \equiv 1$. In this case the numerical predicates are always definable \cite{ImmermanBuch}.
\end{remark}

Using Lemma \ref{lemAppFOint} we can now establish the desired result in the $\FO[\PLUS, \TIMES]$-uniform setting.

\begin{theorem}\label{thm:mainUniform}
$\FO[\PLUS,\TIMES]\textrm{-uniform } \NumAC = \WinFOPT$.
\end{theorem}

\begin{proof}
\underline{$\subseteq$}: Let $f \in \FO\textrm{-uniform } \NumAC$ via the circuit family $\mathcal{C} = (C_n)_{n \in \mathbb{N}}$ and the FO-interpretation $I$ showing its uniformity. With the formula $\varphi$ from the proof of $\NumAC \subseteq \WinFOarb$ this means we have \fa $w$:
\begin{align*}
f(w) & = \textrm{number of proof trees of } C_{|w|} \textrm{ on input } w\\
& = \CWin(\varphi, \mathcal{A}_w \cup C_{|w|}),
\end{align*}
By $\mathcal{A}_w \cup C_{|w|}$ we mean the structure $I(\mathcal{A}_w)$ which is the circuit $C_{|w|}$, where the gates are tuples over the universe of $\mathcal{A}_w$, modified with additional access to the structure  $\mathcal{A}_w$ on the adequate subset of the universe. Let $I'$ be an \FO-interpretation with $I'(\mathcal{A}_w) = \mathcal{A}_w \cup C_{|w|}$ for all $w$. This can easily be constructed from $I$. Now, from $\varphi$ and $I'$ by Lemma \ref{lemAppFOint} we get $\varphi'$ over vocabulary \tString \stfa $\mathcal{A}_w \in \struc[\tString]$:
\begin{align*}
\CWin(\varphi', \mathcal{A}_w) & = \CWin(\varphi, \underbrace{I'(\mathcal{A}_w)}_{\mathcal{A}_w \cup C_{|w|}})\\
& = f(w)
\end{align*}
\underline{$\supseteq$}: We can prove this analogously to $\NumAC \supseteq \WinFOarb$. The only difference is that we need to show FO-uniformity of the circuit. Let $f$ be a function in \WinFO with witness $\varphi$. We now need the formulae $\varphi_\textrm{universe}$, $\varphi_{G_\land}$, $\varphi_{G_\lor}$, $\varphi_E$, $\varphi_\textrm{Input}$, $\varphi_\textrm{negatedInput}$ and $\varphi_r$ defining the circuit.

The second part of the circuit from the proof of the non-uniform version is fixed except for the inputs and constants, which depend on the input size. The inputs are defined by giving the index of the input bit they are associated with. This index is for each input given by one of the variables chosen on the way from the root to the input gate. As we will see in the construction of the first part of the circuit, we will have direct access to these values: Each choice will be stored in a seperate variable. The constants depend on the truth value of certain predicates, so their truth-value is \FO-definable. We can use this by replacing each such constant by a gate with inputs $y$ and $\neg y$ for some variable $y$. The gate is an $\land$-gate, if the constant should be 0 and an $\lor$-gate otherwise.

Now it remains to show, that the first part of the circuit can be made uniform as well. For this, start with an FO-interpretation mapping each string $w \in \struc[\tString]$ to the string $0w \in \struc[\tString]$ (the 0 is appended as the new MSB). Together with the fact that the composition of FO-interpretations is still an FO-interpretation \cite{ImmermanBuch}, it now suffices to give an FO-interpretation
\begin{align*}
\struc[\tString] & \rightarrow \struc[\tCirc]\\
0w & \mapsto C_{|w|}
\end{align*}

Let $\varphi = Q_1 z_1 \dots Q_k z_k \psi(z_1, \dots, z_k)$ with quantifier-free $\psi$. Then each gate in the circuit can be given as a sequence of $k$ values in the range $\{0,\dots,n\}$: The root is the sequence consisting of $n$ in every component. For other nodes, each position in the sequence chooses a child in one level of the circuit. For inner nodes we leave a suffix of a tuple set to $n$, meaning that no children were chosen on those levels. As we can see, having the additional element $n$ as a padding element is quite convenient. 

Following the idea above, the formulae in the \FO-interpretation have to express:
\begin{itemize}
  \item A tuple is in the universe, iff after a component which is set to $n$ all components afterwards have to be $n$ as well. 
  \item There is a directed edge from tuple $\tu{x}$ to tuple $\tu{y}$, iff $\tu{y}$ has exactly one component less set to $n$ than $\tu{x}$.
  \item A gate is an $\land$-gate (resp. $\lor$-gate), iff its layer in the circuit corresponds to a $\forall$-quantifier (resp. $\exists$-quantifier) in $\varphi$.
  \item A gate is the root, if all of its components are set to $n$.
\end{itemize}
In detail, the formulae in the \FO-interpretation can be given as follows, with $\tu{x} = (x_1, \dots, x_k)$ and $\tu{y} = (y_1, \dots, y_k)$:
\[\varphi_\textrm{uni}(\tu{x}) \defeq \quad \bigwedge_{1 \leq i \leq k}\left(x_i = n \to \bigwedge_{i \leq j \leq k}{x_j = n}\right),\]
\begin{align*}
\varphi_E(\tu{x}, \tu{y}) \defeq & \quad \bigwedge_{1\leq i \leq k}{(x_i \neq n \rightarrow x_i = y_i)} \ \land \left(\vphantom{\bigvee_{2 \leq i \leq k}} (y_1 \neq n \land y_2 = n \land x_1 = n) \lor \right.\\
& \quad \left. \phantom{\bigwedge_{1\leq i \leq k}} \bigvee_{2 \leq i \leq k} \big( y_i \neq n \land y_{i+1} = n \land x_i = n \land x_{i-1} \neq n \big) \vphantom{\bigvee_{2 \leq i \leq k}} \right),
\end{align*}
\[\varphi_{G_\land}(\tu{x}) \defeq \quad \bigvee_{\substack{0 \leq i < k\\Q_{i+1} = \forall}}{\left(\bigwedge_{1\leq j \leq i}{x_i \neq n} \land \bigwedge_{i+1 \leq j \leq k}{x_i = n}\right)},\]
\[\varphi_{G_\lor}(\tu{x}) \defeq \quad \bigvee_{\substack{0 \leq i < k\\Q_{i+1} = \exists}}{\left(\bigwedge_{1\leq j \leq i}{x_i \neq n} \land \bigwedge_{i+1 \leq j \leq k}{x_i = n}\right)},\]
\[\varphi_r(\tu{x}) \defeq \quad x_1 = n.\]

\end{proof}

\section{A Model-Theoretic Characterization of \TC}
\label{sect:threshold}

We will now introduce the oracle class $\AC^\NumAC$ as well as \FOWinFOarb, which is a variant of \FO with counting. From the known connections between \TC and \NumAC and from the new connection between \NumAC and \WinFOarb we will then get a new model theoretic characterization of \TC, the class of all languages accepted by Boolean circuits of polynomial size and constant depth with unbounded fan-in AND, OR, and MAJORITY gates, see \cite{HeribertBuch}. Since MAJORITY is an additional type of gate, we have to change the vocabulary slightly. The vocabulary for majority circuits is
\[\tau_{\textrm{maj-circ}} \defeq (E^2, G_\land^1, G_\lor^1, G_{\textrm{MAJ}}^1, \textrm{Input}^2, \textrm{negatedInput}^2, r^1),\]
All predicates shared with $\tCirc$ are interpreted in the same way as before. The meaning of the new predicate is as follows:
\begin{itemize}
  \item $G_\textrm{MAJ}(x)$: gate $x$ is a MAJORITY gate
\end{itemize}

We now want to define oracle classes. A Boolean oracle-circuit with a function-oracle is a circuit with additional oracle gates we will call $\#$-gates. Within a Boolean circuit we obviously can not allow arbitrary outputs of oracle functions. Therefore, we allow only bitwise access to the oracle and $\#$-gates are labeled with an index to indicate which Bit they output. This means, if $f$ is the oracle and we have an $\#$-gate labeled with $i$, it computes the Boolean function
\begin{alignat*}{2}
f_i\colon & \{0,1\}^* && \rightarrow \{0,1\}\\
& x && \mapsto \textrm{BIT}(i, f(x)).
\end{alignat*}

Additionally, in contrast to the Boolean operations $\land$, $\lor$ and MAJORITY, the oracle gates do not neccessarily compute commutative functions, meaning that we need to specify the order of the inputs. We take this a step further and give indices for the inputs to each oracle gate. This leads us to the following vocabulary for oracle-circuits:
\[\tau_\textrm{o-circ} \defeq (E^2, G_\land^1, G_\lor^1, G_\#^2, \textrm{Index}^3, \textrm{Input}^2, \textrm{negatedInput}^2, r^1)\]
The meaning of the predicates it shares with $\tCirc$ stays the same. The meaning of the new predicates is as follows:
\begin{itemize}
  \item $G_\#(x,i)$: gate $x$ is an oracle gate accessing the $i$-th bit of the oracle function
  \item $\textrm{Index}(x,y,i)$: $x$ is an oracle gate and $y$ is the $i$-th child of $x$ 
\end{itemize}

We use the following notation: Let $C$ be a class of Boolean circuits, $A$ the class of languages decidable by circuits from $C$ and $B$ a class of functions $\{0,1\}^* \rightarrow \{0,1\}^*$. Then $A^B$ is the class of languages decidable by circuits from $C$ with the addition of oracle gates, using oracles from $B$. For this purpose, functions $\{0,1\}^* \to \mathbb{N}$ can of course be viewed as functions $\{0,1\}^* \to \{0,1\}^*$ by encoding natural numbers in binary.

Additionally, we will also use oracle majority circuits. For this we use the vocabulary
\[\tau_\textrm{o-maj-circ} \defeq (E^2, G_\land^1, G_\lor^1, G_\textrm{MAJ}^1, G_\#^2, \textrm{Index}^3, \textrm{Input}^2, \textrm{negatedInput}^2, r^1),\]
where the predicates are interpreted analogously to their meaning in $\tau_\textrm{maj-circ}$ or $\tau_\textrm{o-circ}$, respectively.

The main result of this section will be a new characterization of the circuit class $\TC$ using a certain two-sorted logic.

\begin{definition}\label{def:FOCW}
Given a vocabulary $\sigma$, a $\sigma$-structure for \FOWinFOarb is a structure of the form
\[(\{a_0, \dots, a_{n-1}\}, \{0, \dots, n-1\}, (R_i)^\mathcal{A}, +, \times, \underline{\min}, \underline{\max}),\]
where $(\{a_0, \dots, a_{n-1}\}, (R_i)^\mathcal{A}) \in \textrm{STRUC}[\sigma]$, $+$ and $\times$ are the ternary relations corresponding to addition and multiplication in $\mathbb{N}$ and $\underline{\min}, \underline{\max}$ denote $0$ and $n-1$, respectively. We assume that the two universes are disjoint. Formulae can have free variables of two sorts.\\
This logic extends the syntax of first order logic as follows:
\begin{itemize}
  \item terms of the second sort:
    $\underline{\min}$, $\underline{\max}$
  \item formulae:
  \begin{enumerate}[(1)]
    \item if $t_1, t_2, t_3$ are terms of the second sort, then the following are (atomic) formulae: $+(t_1, t_2, t_3), \times(t_1, t_2, t_3)$
    \item if $\varphi(x,i)$ is a formula, then also $\exists i \varphi(x,i)$ (binding the second-sort variable $i$)
    \item if $Q$ is a quantifier prefix quantifying the first-sort variables $\overline{x}$ and the second-sort variables $\overline{i}$, $\varphi(\overline{x}, \overline{i})$ is a quantifier-free formula and $\overline{j}$ a tuple of second-sort variables, then the following is a formula: $\#_{Q\varphi}(\overline{j})$
  \end{enumerate}
\end{itemize}
\end{definition}
The semantics is clear except for $\#_{Q\varphi}(\overline{j})$. Let $\mathcal{A}$ be an input structure and $\overline{j}_0$ an assignment for $\overline{j}$. Then
\[\mathcal{A} \vDash \#_{Q\varphi}(\overline{j}_0) \quad \defeqv \quad \textrm{the val($\overline{j}_0$)-th bit of } \CWin(Q\varphi, \mathcal{A}) \textrm{ is } 1\]
Here, $\textrm{val}(\overline{j}_0)$ denotes the numeric value of the vector $\overline{j}_0$ under an appropriate encoding of the natural numbers as tuples of elements from the second sort.

The types (1) and (2) of formulae in our definition are the same as in Definition~8.1 in \cite[p.~142]{LibkinBuch}. Additionally, our definition allows new formulae $\#_{Q\varphi}(\overline{j})$. These allow us to talk about the number of winning strategies for sub-formulae $Q\varphi$. Note that these numbers can be exponentially large, hence polynomially long in binary representation; therefore we can only talk about them using some form of BIT predicate. Formulae of type (3) are exactly this: a BIT predicate applied to a number of winning strategies.

Our logic \FOWinFOarb thus gives \FO with access to number of winning strategies, i.e., in \FOWinFOarb we can count in an exponential range. Libkin's logic FO(Cnt)$_{\mathsf{All}}$ can count in the range of input positions, i.e., in a linear range. Nevertheless we will obtain the maybe somewhat surprising result that both logics are equally expressive on finite structures: both correspond to the circuit class \TC.

First, we give a technical result showing a certain closure property of $\NumAC$.

\begin{definition}
A class $\mathcal{C}$ of functions $\{0,1\}^* \to \mathbb{N}$ is closed under polynomially padded concatenation if for all $f,g \in \mathcal{C}$ there is a polynomial $p$ \ST the function
\[h(x) \defeq f(x)0^{p(|x|)-|g(x)|}g(x)\]
is also in $\mathcal{C}$, where by $f(x)0^{p(|x|)-|g(x)|}g(x)$ we mean the concatenation of the parts as a binary string.
\end{definition}

\begin{lemma}\label{lem:polyPad}
\NumAC and $\FO[\PLUS, \TIMES]$-uniform \NumAC are closed under polynomially padded concatenation. 
\end{lemma}

\begin{proof}
Let $f, g \in \NumAC$ via circuit families $\mathcal{C}_1, \mathcal{C}_2$ and let $p$ be a polynomial bounding the length of all outputs of the function $g$ depending on the length of the input. The following circuit shows $h(x) \defeq f(x)0^{p(|x|) - |g(x)|}g(x) \in \NumAC$:\\
Compute $f(x)$ using the circuit $\mathcal{C}_1$, then shift the result by $p(|x|)$ bits to the left. Compute seperately $g(x)$ using the circuit $\mathcal{C}_2$. Add both results together. Shifting can be done by multiplication with a subcircuit computing $2^{p(|x|)}$. Figure \ref{tikzPic} illustrates a bit more detailed how this is done with our definition of a circuit (edges have no multiplicities).

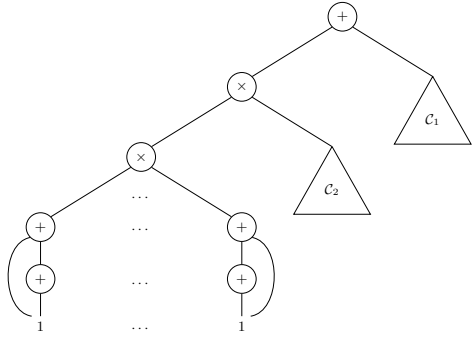
\begin{figure}
\capstart
\begin{center}
\scalebox{.6}{
\begin{tikzpicture}

\node [circle, draw](root) {$+$};
\draw (root.south west) -- ($(root.south) - (2,1)$);
\draw (root.south east) -- ($(root.south) + (2,-1)$) node (C1) {} -- ($(C1) - (0.85,1.5)$) -- ($(C1) + (0.85,-1.5)$) -- ($(root.south) + (2,-1)$);
\node[below=6mm of C1] {$\mathcal{C}_1$};

\node [circle, draw] (shift)[below left= 1 and 2 of root, anchor=north] {$\times$};
\draw (shift.south west) -- ($(shift.south) - (2,1)$);
\draw (shift.south east) -- ($(shift.south) + (2,-1)$) node (C2) {};
\draw ($(shift.south) + (2,-1)$) -- ($(C2) - (0.85,1.5)$) -- ($(C2) + (0.85,-1.5)$) -- ($(shift.south) + (2,-1)$);
\node[below=6mm of C2] {$\mathcal{C}_2$};

\node [circle, draw] (pwr2)[below left= 1 and 2 of shift, anchor=north] {$\times$};
\draw (pwr2.south west) -- ($(pwr2.south) - (2,1)$);
\draw (pwr2.south east) -- ($(pwr2.south) + (2,-1)$);
\node[below=4mm of pwr2] {$\cdots$};

\node [circle, draw] (plus1)[below left= 1 and 2 of pwr2, anchor=north] {$+$};
\node [circle, draw] (plus2)[below right= 1 and 2 of pwr2, anchor=north] {$+$};
\draw (plus1.south) -- ($(plus1.south) - (0,0.5)$) node[circle, draw, anchor=north] (plus3) {$+$};
\draw (plus2.south) -- ($(plus2.south) - (0,0.5)$) node[circle, draw, anchor=north] (plus4) {$+$};
\draw (plus3.south) -- ($(plus3.south) - (0,0.5)$) node[anchor=north] (one1) {1};
\draw (plus4.south) -- ($(plus4.south) - (0,0.5)$) node[anchor=north] (one2) {1};

\node[below=11mm of pwr2] {$\cdots$};
\node[below=23mm of pwr2] {$\cdots$};
\node[below=33mm of pwr2] {$\cdots$};

\draw (plus1.south west) to[bend right=80] (one1.north west);
\draw (plus2.south east) to[bend left=80] (one2.north east);

\end{tikzpicture}
}
\caption{\NumAC-circuit for $f(x)0^{p(|x|) - |g(x)|} g(x)$}
\label{tikzPic}
\end{center}
\end{figure}

Obviously, this circuit has still polynomial size in the input length and computes $h(x)$.\\
It also can be easily seen that if $\mathcal{C}_1$ and $\mathcal{C}_2$ are uniform circuit families, then the constructed circuit family is uniform.
\end{proof}

\begin{theorem}\label{thm:threshold}
$\TC = \FOWinFOarb = {\AC}^{\NumAC}$
\end{theorem}

\begin{proof}
A central ingredient of the proof is the known equality $\TC=\textrm{PAC}^0$ from \cite{PAC0TC0}. Here, $\textrm{PAC}^0$ is defined to be the class of languages $L$ for which there exist functions $f,h\in\NumAC$ such that for all $x$, $x\in L$ iff $f(x)>h(x)$.
We prove the result by establishing the following chain of inclusions:
\begin{align*}
\TC & \subseteq \textrm{PAC}^0\\
	 & \stackrel{(1)}{\subseteq} \FOWinFOarb\\
	 &  \stackrel{(2)}{\subseteq} {\textrm{AC}^0}^{\textrm{\#AC}^0}\\
	 & \stackrel{(3)}{\subseteq} \TC
\end{align*}
As mentioned above, $\TC \subseteq \textrm{PAC}^0$ is known. We will now show the rest of the above inclusions one by one.

\textit{Proof of (1)}: Let $L \in \textrm{PAC}^0$. There are $f,g \in \textrm{\#AC}^0$ \stfa $x$:
\[x \in L \Leftrightarrow f(x) > g(x)\]
Since $\textrm{\#AC}^0 = \WinFOarb$, we have $f,g \in \WinFOarb$.  Therefore, we can get the bits of $f(x)$ and $g(x)$ using the \#-predicate. Let $\psi_1, \psi_2$ be the formulae which show $f \in \WinFOarb$ and $g \in \WinFOarb$, respectively. We now need a formula checking whether $f(x) > g(x)$. We know that $|f(x)|$ and $|g(x)|$ are both polynomially bounded in $|x|$. This means, that we need tuples of variables of the second sort as indices. The formula roughly looks as follows:
\begin{align*}
\exists \tu{p}_0 & \Big[ f_{\tu{p}_0}(x) = 1 \land \forall(\tu{q} > \tu{p}_0) (f_{\tu{q}}(x) = 0 \land g_{\tu{q}}(x) = 0)\\
& \land \exists (\tu{p}_1 \leq \tu{p}_0) \big((\forall(\tu{p}_0 \geq \tu{q} > \tu{p}_1) (f_{\tu{q}}(x) = g_{\tu{q}}(x))\big) \land f_{\tu{p}_1}(x) > g_{\tu{p}_1}(x)\Big]
\end{align*}
Here, $f_i$ (resp. $g_i$) is a shortcut for $\#_{\psi_1}(i)$ (resp. $\#_{\psi_2}(i)$). Note that all variables in the formula are second sort variables.

\textit{Proof of (2)}: Let $L \in \textrm{FOCW[Arb]}$ via $\varphi$, where $\varphi$ is in prenex normal form. Second sort variables can be treated in the same way as first sort variables. This means that the only difference to a usual $\FO$-formula are occuring $\#$-predicates. Hence, we can build a circuit from $\varphi$ analogously to the one in the proof for $\FO[\textrm{Arb}] \subseteq \AC$. The difference is, that within the second part of the circuit, evaluating the quantifier-free part of $\varphi$, some $\#$-predicates have to be evaluated. If all occurrences of the $\#$-predicate in $\varphi$ access bits of the same \NumAC-function, this would be easy: We would use that function as our oracle and would need exactly one oracle gate for each occurrence. If the $\#$-predicates refer to different $\AC$-functions, we first need to build a single oracle from them, which allows to access each of them as needed. For this, we use a polynomially padded concatenation: Since there are only constantly many different \NumAC-functions accessed via occurrences of $\#$-predicates in $\varphi$, polynomially padded concatenation of all of them is still in \NumAC. Now again, we can use exactly one oracle gate to compute the value of each of the occurrences---only the index must be changed according to the construction of the polynomially padded concatenation.

\textit{Proof of (3)}: This can be seen with the following chain of inclusions:
\begin{align*}
{\AC}^\NumAC & \stackrel{\textrm{(i)}}{\subseteq} {\TC}^{\NumAC}\\
& \stackrel{\textrm{(ii)}}{\subseteq} {\TC}^{\TC}\\
& \stackrel{\textrm{(iii)}}{\subseteq} \TC
\end{align*}
(i) follows from the fact that \AC oracle-circuit families are also \TC oracle-circuit families. (ii) is due to ITADD and ITMULT being in \TC. This means that the bits of \NumAC-functions can be computed in \TC. Since function oracles only give bitwise access, a \NumAC-oracle can be replaced by a \TC-oracle that takes the output index as additional input. For (iii), take any \TC oracle-circuit family and \TC-oracle. A \TC circuit family computing the same function can be constructed by replacing all oracle gates by subcircuits computing the needed function.

\end{proof}


We now want to show the uniform version of Theorem \ref{thm:threshold}. We start by showing that our definition of oracle-circuits in fact allows us to replace oracle gates by adequate subcircuits, as long as the oracle is weak enough. We do this for the special case of the class \TC and also for a language oracle instead of a function oracle. A language oracle for language $L$ can be viewed as a function oracle of the characteristic function $c_L$, such that we do not need a new definition.

\begin{lemma}\label{lem:oracle}
Let $A \subseteq \TC$. Then ${\TC}^A \subseteq \TC$.
\end{lemma}

\begin{proof}
Let $L \in {\TC}^A$ via the \TC oracle-circuit family $\mathcal{C}$ and oracle $L' \in A$. Let $\mathcal{D}$ be a \TC circuit family deciding $L$. We now construct a \TC circuit family deciding $L$. Intuitively, this can be done by replacing the oracle gates by subcircuits deciding $L'$. These can be taken from $\mathcal{D}$. We now make this formal. The first thing to notice is that the number of inputs of different oracle gates within circuits from $\mathcal{C}$ can differ. Let $q$ be the polynomial bounding the size of circuits in $\mathcal{C}$ depening on the input length. Then the number of inputs to all oracle gates is also bounded by $q$. There is also a polynomial $p$ bounding the size of circuits in $\mathcal{D}$ depending on the input length. Thus, when replacing oracle gates by subcircuits from $\mathcal{D}$, the size of each subcircuit is bounded by $p \circ q$. This allows us to fix the tuple length for representation of gates in the new circuit to be the old tuple length plus $\textrm{deg}(p \circ q)$ (and possibly some part to deal with structures with small universes). Similarly to what we do in the proof of Theorem \ref{thm:mainUniform}, we add an  additional element $n$ to the universe and use it as padding. For all non-oracle gates only the first part of the tuple is used and connections are only built if all newly added components are $n$. For oracle gates, the second part is used to represent gates within the subcircuits. Gates that were parents of the oracle gate before will be parents of the output gate of the corresponding subcircuit afterwards. Gates that were children of the oracle gate will be directly connected to the input gates of the corresponding subcircuit (and these inputs can be made $\land$-gates). For the inner connections, the uniformity of $\mathcal{D}$ is used: Let $\mathcal{D} = (D_n)_{n \in \mathbb{N}}$. We know that there is an \FO-interpretation mapping $\mathcal{A}_w \mapsto D_{|w|}$. Due to the Index-predicate it is possible to determine the number of inputs of a given oracle gate. Let $m$ be this number. Since the numerical predicates can be extended to tuples \cite{ImmermanBuch}, there is also an \FO-interpretation $\mathcal{A}_w \mapsto \mathcal{A}_{0^m}$. Together, we get an \FO-interpretation $\mathcal{A}_w \mapsto D_m$. This means that the oracle gates can be uniformly replaced by subcircuits.
\end{proof}

We are now in the position to proof the uniform version of Theorem \ref{thm:threshold}.

\begin{theorem}
In the $\FO[\PLUS, \TIMES]\textrm{-uniform}$ setting,
\[\TC = \mathrm{FOCW} = {\AC}^\NumAC.\]
\end{theorem}

\begin{proof}
We proof this analogously to Theorem \ref{thm:threshold}. We explain how the constructions can be made uniform, where neccessary. Recall the chain of inclusions from before:
\begin{align*}
\TC & \subseteq \mathrm{PAC}^0\\
& \stackrel{(1)}{\subseteq} \mathrm{FOCW}\\
& \stackrel{(2)}{\subseteq} \AC^\NumAC\\
& \stackrel{(3)}{\subseteq} \TC
\end{align*}
The $\FO[\PLUS, \TIMES]$-uniform version of $\TC \subseteq \mathrm{PAC}^0$ is known from \cite{ABL98}. We will now show that the rest of the above inclusions can be made uniform as well.

\textit{Proof of (1)}: Since $\NumAC = \WinFO$ also holds in the uniform case, we can use the same idea. Also, the formula given in the proof of Theorem \ref{thm:threshold} only uses non-uniformity for the order-relation and can thus be used for the uniform case as well.

\textit{Proof of (2)}: This can be done in the same way as in the non-uniform case. Since Lemma \ref{lem:polyPad} also holds uniformly, the only problem left is choosing the right indices for the oracle gates. Lets first assume, that the formula only uses (possibly multiple bits of) one oracle. In this case, the index is given by a tuple of quantified variables. These are part of the representation of any input gate and can be directly accessed to uniformly describe the index.

If there are two different oracles used within the formula, we use a polynomially padded concatenation of them again. Let $f, g \in \NumAC$ be the oracle functions and $f(x) 0^{p(|x|) - |g(x)|} g(x|)$ their polynomially padded concatenation. All oracle gates that query $g$ can use the same index as in the original formula. All oracle gates that query $f$ have to add $p(|x|)$ to the index. This can be easily done in $\FO[\PLUS, \TIMES]$, so the index can be described uniformly. Now, for an arbitrary number of $\#$-predicates, the above can be applied inductively.

\textit{Proof of (3)}: This can be seen with the same chain of inclusions as in the non-uniform version:
\begin{align*}
{\AC}^\NumAC & \stackrel{\textrm{(i)}}{\subseteq} {\TC}^{\NumAC}\\
& \stackrel{\textrm{(ii)}}{\subseteq} {\TC}^{\TC}\\
& \stackrel{\textrm{(iii)}}{\subseteq} \TC
\end{align*}
(i) is again trivial. (ii) is possible, since ITADD and ITMULT are in uniform \TC by \cite{HesseDivTC0} and the construction from the non-uniform proof can be made uniform. (iii) is an application of Lemma \ref{lem:oracle}.
\end{proof}

\section{Conclusion}
\label{sect:concl}

Arithmetic classes are of current focal interest in computational complexity, but no model-theoretic characterization for any of these was known so far. We addressed the maybe most basic arithmetic class \NumAC and gave such a characterization, and, based on this, a new characterization of the (Boolean) class \TC. 

This immediately leads to a number of open problems:
\begin{itemize}
\item We mentioned the logical characterization of $\NumP$ in terms of counting assignments to free relations. We here count assignments to free function variables. Hence both characterizations are of a similar spirit. Can this be made more precise? Can our class $\WinFO[\PLUS, \TIMES]$ be placed somewhere in the hierarchy of classes from \cite{DescCompNumP}?
\item Can larger arithmetic classes be defined in similar ways? The next natural candidate might be $\#\NC$ which corresponds to counting paths in so called non-uniform finite automata \cite{camcthvo96}. Maybe this will lead to a descriptive complexity characterization.
\item Still the most important open problem in the area of circuit complexity is the question if $\TC=\NC$. While we cannot come up with a solution to this, it would be interesting to reformulate the question in purely logical terms, maybe making use of our (or some other) logical characterization of $\TC$.
\end{itemize}

\subsection*{Acknowledgements}

We are grateful to Lauri Hella (Tampere) and Juha Kontinen (Helsinki) for helpful discussion, leading in particular to Definition~\ref{def:FOCW}.
We also thank the anonymous referees for helpful comments.

\bibliographystyle{elsarticle-num}
\bibliography{references,cc}

\end{document}